\documentclass[11pt,letterpaper]{article}

\usepackage{amsthm,amsmath,amssymb,color,fullpage,cleveref,cite}
\usepackage{mathrsfs}
\usepackage[ruled]{algorithm2e}
\usepackage[margin=1in]{geometry}
\usepackage[singlespacing]{setspace}

\newtheorem{theorem}{Theorem}[section]
\newtheorem{lemma}[theorem]{Lemma}

\newtheorem{corollary}[theorem]{Corollary}
\newtheorem{claim}[theorem]{Claim}
\newtheorem{fact}[theorem]{Fact}
\newtheorem{conjecture}[theorem]{Conjecture}
\newtheorem{definition}[theorem]{Definition}

\newtheorem{example}[theorem]{Example}

\newcommand{\eps}{\varepsilon}
\newcommand{\abs}[1]{\left| #1 \right|}
\newcommand{\pbra}[1]{\left( #1 \right)}
\newcommand{\sbra}[1]{\left[ #1 \right]}
\newcommand{\cbra}[1]{\left\{ #1 \right\}}
\newcommand{\bin}{\{0,1\}}
\newcommand{\poly}{\text{poly}}
\newcommand{\Uscr}{\mathscr{U}}

\newcommand{\Rscr}{\mathscr{R}}
\newcommand{\Ecal}{\mathcal{E}}
\newcommand{\Ncal}{\mathcal{N}}
\newcommand{\Ucal}{\mathcal{U}}
\newcommand{\Vcal}{\mathcal{V}}
\newcommand{\usenum}[1]{\#\text{useful}\pbra{#1}}

\DeclareMathOperator*{\E}{\mathbb E}
\newcommand{\Ind}{\text{Ind}}
\newcommand{\Stab}{\text{Stab}}
\newcommand{\Enc}{\text{Enc}}
\newcommand{\Dec}{\text{Dec}}
\newcommand{\True}{\texttt{True}}

\title{Decision list compression by mild random restrictions}
\author{
Shachar Lovett\thanks{Research supported by NSF award 1614023.}\\
Computer Science Department\\
University of California, San Diego\\
\texttt{shachar.lovett@gmail.com}
\and
Kewen Wu\\
School of EECS\\
Peking University, Beijing\\
\texttt{shlw\_kevin@pku.edu.cn}
\and
Jiapeng Zhang\thanks{Research Supported by NSF grant CCF-1763299 and Salil Vadhan's Simons Investigator Award.}\\
School of Engineering and Applied Science\\
Harvard University\\
\texttt{jpeng.zhang@gmail.com}
}

\begin{document}

\maketitle

\begin{abstract}
A decision list is an ordered list of rules. Each rule is specified by a term, which is a conjunction of literals, and a value. Given an input, the output of a decision list is the value corresponding to the first rule whose term is satisfied by the input. Decision lists generalize both CNFs and DNFs, and have been studied both in complexity theory and in learning theory.

The size of a decision list is the number of rules, and its width is the maximal number of variables in a term. We prove that decision lists of small width can always be approximated by decision lists of small size, where we obtain sharp bounds. This in particular resolves a conjecture of Gopalan, Meka and Reingold (Computational Complexity, 2013) on DNF sparsification.

An ingredient in our proof is a new random restriction lemma, which allows to analyze how DNFs (and more generally, decision lists) simplify if a small fraction of the variables are fixed. This is in contrast to the more commonly used switching lemma, which requires most of the variables to be fixed.
\end{abstract}

\section{Introduction}
\label{sec:intro}

Decision lists are a model to represent boolean functions, first introduced by Rivest \cite{rivest1987learning}. A decision list is given by a list of rules $(C_1,v_1),\ldots,(C_m,v_m)$. A rule is composed of a condition, given by a term $C_i$, which is a conjunction of literals (variables or their negations); and an output value $v_i$ in some set $V$. A decision list computes a function $f:\{0,1\}^n \to V$ as follows:
\begin{center}
\begin{minipage}{0.5\textwidth}
\textbf{If $C_1(x)=\True$ then} output $v_1$,\\
\textbf{else if $C_2(x)=\True$ then} output $v_2$,\\
\ldots, \\
\textbf{else if $C_{m}(x)=\True$ then} output $v_{m}$.
\end{minipage}
\end{center}
The last rule is the \emph{default value}, where we assume that $C_m \equiv \True$.

Decision lists generalize both CNFs and DNFs. For example, a DNF is a decision list with $v_1=\cdots=v_{m-1}=1$ and $v_m=0$, and a CNF is a decision list with $v_1=\cdots=v_{m-1}=0$ and $v_m=1$. It can be shown that decision lists are a strict generalization of both DNFs and CNFs \cite{rivest1987learning,kohavi1993research}. Following Rivest's original work, decision lists have been studied both in complexity theory \cite{blum1992rank,turan1997linear,guijarro2001monotone,eiter2002decision,krause2006computational,arvind2015isomorphism,abs-1901-05911} and in learning theory \cite{kearns1987learnability,bagallo1990boolean,hancock1996lower,ehrenfeucht1989general,nevo2002online,wang2015falling,witten2016data}.

\paragraph{Complexity measures of decision lists.}
There are two natural complexity measures of decision lists: \emph{size} and \emph{width}. Let $L=((C_i,v_i))_{i\in[m]}$ be a decision list. Its \emph{size} is the number of rules in it (namely $m$), and its \emph{width} is the maximal number of variables in a term $C_i$.

\paragraph{Decision list approximation.}
A decision list $L$ $\eps$-approximates another decision list $L'$ if the two agree on a $(1-\eps)$ fraction of the inputs.
It is straightforward to see that small-size decision lists can be approximated by small-width decision lists, by removing rules of large width. Concretely, a decision list of size $m$ can be $\eps$-approximated by a decision list of width $w=\log(m/\eps)$, simply by removing all rules with terms of width more than $w$.
The reverse direction is the main focus of this work. We prove the following result, which provides sharp bounds on approximating small-width decision lists by small-size decision lists.

\begin{theorem}[Main result]\label{thm:main}
Let $w \ge 1, \eps>0$.
Any width-$w$ decision list $L$ can be $\eps$-approximated by a decision list $L'$ of width $w$ and size $s=\pbra{2+\frac1w\log\frac1\eps}^{O(w)}$. Moreover, $L'$ is a sub-decision list of $L$, obtained by keeping $s$ rules in $L$ and removing the rest. The bound on $s$ is optimal, up to the unspecified constant in the $O(w)$ term.
\end{theorem}

The proof of \Cref{thm:main} appears in \Cref{sec:upper}.
We note that the size bound can be simplified, depending on whether
the required error $\eps$ is below or above $2^{-w}$:
$$
\pbra{2+\frac1w\log\frac1\eps}^{O(w)}=\begin{cases}
2^{O(w)}&\eps \ge 2^{-w}\\
\pbra{\frac2w\log\frac1\eps}^{O(w)}&\eps \le 2^{-w}.
\end{cases}
$$
In both cases, the bound we obtain is sharp, up to the unspecified constant in the $O(w)$ term. We give examples demonstrating this in \Cref{sec:lower}.

\subsection{Random restrictions}

Random restrictions are an essential ingredient of the proof of \Cref{thm:main}.
H{\aa}stad's switching lemma \cite{Hastad87,razborov1995bounded,beame1994switching} is based on the fact that small-width DNFs simplify under random restrictions. More concretely, a random restriction that fixes a $1-O(1/w)$ fraction of the inputs simplifies a width-$w$ DNF to a small-depth decision tree. In this work, we study random restrictions where a small constant fraction of the variables is fixed.

A good example to keep in mind is the TRIBES function: a read-once DNF with $2^w$ terms of width $w$ on disjoint variables. The TRIBES function does not simplify significantly under a random restriction, unless one really fixes a  $1-O(1/w)$ fraction of the inputs. For example, if we randomly fix $50\%$ of the inputs, say, then the TRIBES function  simplifies to what is essentially a smaller TRIBES function (more formally, it simplifies with high probability to a read-once DNF of width $\Omega(w)$). However, we show that this is in essence the worst possible example.

The following lemma is a special case of \Cref{lemma:expectedusenum} applied to DNFs (the full lemma deals with decision lists). Given a DNF $f:\bin^n \to \bin$, let $\rho \in \{0,1,*\}^n$ be a restriction, and let $f\restriction_{\rho}$ be the restricted DNF. Clearly, some terms in $f$ might become redundant in $f\restriction_{\rho}$. For example, they could be false, or they could be implied by other terms. A term that is not redundant is called \emph{useful}. We show that after fixing even a small fraction of the variables (say, $1\%$), a width-$w$ DNF simplifies to have at most $2^{O(w)}$ useful terms, and hence cannot be ``too complicated''.

\begin{lemma}[DNFs simplify after mild random restrictions]
\label{lemma:dnf_useful}
Let $f$ be a width-$w$ DNF, and let $f\restriction_{\rho}$ be a restriction of $f$ obtained by restricting each variable with probability $\beta$, where the restricted variables take values 0 and 1 with equal probability. Then the expected number of useful terms in $f\restriction_{\rho}$ is at most $(4/\beta)^w$.
\end{lemma}

\subsection{Applications}

We discuss some applications of \Cref{thm:main} below.

\subsubsection{DNF sparsification}

This decision list compression problem is a natural generalization of the \emph{DNF sparsification} problem, introduced by Gopalan, Meka and Reingold \cite{GopalanMR13} as a means to obtain pseudorandom generators fooling small-width DNFs. Their main structural result can be summarized as follows.

\begin{theorem}[\cite{GopalanMR13}]
\label{thm:GMR}
Any width-$w$ DNF can be $\eps$-approximated by a DNF of width $w$ and size $(w\log(1/\eps))^{O(w)}$.
\end{theorem}

They conjectured that a better bound is possible.

\begin{conjecture}[\cite{GopalanMR13}]
\label{conj:GMR}
Any width-$w$ DNF can be $\eps$-approximated by a DNF of width $w$ and size $s(w,\eps)$, where:
\begin{itemize}
\item \emph{Weak version:} $s(w,\eps)=c(\eps)^w$ for some function $c$.
\item \emph{Strong version:} $s(w,\eps)=(\log(1/\eps))^{O(w)}$.
\end{itemize}
\end{conjecture}

The weak version was resolved by Lovett and Zhang \cite{LovettZ19}, where they showed that $c(\eps)=(1/\eps)^{O(1)}$ suffices. Our main result, \Cref{thm:main}, verifies the strong version of their conjecture (and in fact, proves a sharper bound than the one conjectured).

\begin{corollary}[This work]\label{cor:dnfsparsification}
Any width-$w$ DNF can be $\eps$-approximated by a DNF of width $w$ and size $\pbra{2+\frac1w\log\frac1\eps}^{O(w)}$.
\end{corollary}

We remark that \Cref{cor:dnfsparsification} is also tight, up to the unspecified constant in the $O(w)$ term. The proof is very similar to the proof in \Cref{sec:lower} that \Cref{thm:main} is tight. We sketch the proof here:
\begin{itemize}
\item For $2^{-2w}\leq\eps\leq1/3$, \Cref{lem:approxlowerbound} shows the existence of a function $f:\{0,1\}^w \to \{0,1\}$ that cannot be $(1/3)$-approximated by any decision list of width $w$ and size $O(2^w/w)$. In particular, $f$ cannot be approximated by a DNF of width $w$ and size $O(2^w / w)$. Note that $f$ can trivially be computed by a DNF of width $w$ and size $2^w$, and that 
$2^{\Omega(w)}=\pbra{2+\frac1w\log\frac1\eps}^{\Omega(w)}$ in this regime.
\item For $\eps\leq2^{-2w}$, consider exactly computing the Threshold-$w$ function on $\log(1/\eps)$ variables, which amounts to approximation with any error $<\eps$.
This requires a width-$w$ DNF of size $\binom{\log(1/\eps)}w=\pbra{2+\frac1w\log\frac1\eps}^{\Omega(w)}$.
\end{itemize}

\subsubsection{Junta theorem}

A $k$-junta is a function depending on at most $k$ variables. Friedgut's junta theorem \cite{friedgut1998boolean} shows that boolean functions of small influence can be approximated by juntas. For the relevant definitions see for example \cite{o2014analysis}.

\begin{theorem}[Friedgut's junta theorem \cite{friedgut1998boolean}]
\label{thm:friedgut}
Let $f:\bin^n \to \bin$ be a boolean function with total influence $I$. Then for any $\eps>0$, $f$ can be $\eps$-approximated by a $k$-junta for $k=2^{O(I/\eps)}$.
\end{theorem}

It is well known that width-$w$ DNFs have total influence $I=O(w)$, which implies by \Cref{thm:friedgut} that width-$w$ DNFs can be $\eps$-approximated by $2^{O(w/\eps)}$-juntas. Since a width-$w$ size-$s$ decision list is a $(sw)$-junta, as a corollary of \Cref{thm:main}, we improve the bound, and generalize it to decision lists.

\begin{corollary}[This work]
Any width-$w$ decision list can be $\eps$-approximated by a $k$-junta for $k=\pbra{2+\frac1w\log\frac1\eps}^{O(w)}$.
\end{corollary}

This improves previous bounds, even when restricted to DNFs or CNFs. By combining the results in \cite{GopalanMR13,LovettZ19} one gets the bound $k=\min\cbra{w\log(1/\eps),1/\eps}^{O(w)}$ for width-$w$ DNFs or CNFs. It can be verified that our new result is indeed better; for example for $\eps=w^{-w}$ we obtain $(\log w)^{O(w)}$ instead of $w^{O(w)}$. It is also worthwhile noting that the result of \cite{LovettZ19}, which obtained the bound $(1/\eps)^{O(w)}$, can be extended to decision lists with minimal changes.

\subsubsection{Learning small-width DNFs}
A class of boolean functions is said to be $(\eps,\delta)$-PAC learnable using $q$ queries if there exists a learning algorithm that, given query access to an unknown function in the class, returns with probability $(1-\delta)$ a function which $\eps$-approximates the unknown function, while making at most $q$ queries. In our context we consider membership queries, where the learning algorithm can query the value of the unknown function on any chosen input.

A celebrated result of Jacskson \cite{jackson1997efficient} shows that polynomial-size DNFs can be PAC learned under the uniform distribution using membership queries.

\begin{theorem}[Jackson's harmonic sieve \cite{jackson1997efficient}]
The class of $n$-variate DNFs of size $s$ is $(\eps,\delta)$-PAC learnable under the uniform distribution with $q=\poly(s,n,1/\eps,\log(1/\delta))$ membership queries.
\end{theorem}

Using \Cref{thm:main}, we can extend Jackson's result to small-width DNFs. Note that the DNF sparsification bound from \cite{GopalanMR13,LovettZ19} also works here, if we replace the bound on $s$ with their corresponding bound.

\begin{corollary}[This work]
The class of $n$-variate DNFs of width $w$ is $(\eps,\delta)$-PAC learnable under the uniform distribution with $q=\poly(s,n,1/\eps,\log(1/\delta))$ membership queries, where $s=\pbra{2+\frac1w\log\frac1\eps}^{O(w)}$.
\end{corollary}

\begin{proof}[Proof Sketch]
Jackson's algorithm combines a weak learner based on Fourier analysis and a boosting algorithm that converts this weak learner to a strong learner. Let $f(x)$ be the target DNF that we are trying to learn.
The weak learner solves the following problem: given a distribution $D$ on $\bin^n$, output a set $S$ such that the parity $\chi_S(x)=\bigoplus_{i \in S} x_i$ is correlated with $f$ under the distribution $D$. Initially $D$ is the uniform distribution, but the boosting algorithm keeps adapting $D$ to focus on inputs where it made many mistakes.

In Jackson's algorithm, the existence of such $S$ is shown by observing that for a size-$s$ DNF, at least one of the terms must be $1/s$ correlated to the function; and each term's contribution can be attributed to the parities supported on it. For width-$w$ terms, this leads to at most a $2^{-w}$ decrease in the correlation.

Assume now that $f(x)$ is a width-$w$ DNF with too many terms, so we cannot apply the previous argument directly. Apply \Cref{thm:main} with error $\gamma$ (to be determined soon), to obtain an approximate width-$w$ DNF $g(x)$ which $\gamma$-approximates $f(x)$, where $g$ has at most $s=\pbra{2+\frac1w\log\frac1{\gamma}}^{O(w)}$ terms. Crucially, we obtain $g(x)$ by removing some of the terms in $f(x)$, and hence $g(x) \le f(x)$ for all inputs $x$. In particular, $\Pr_{x \sim D}[f(x)=1] \ge \Pr_{x \sim D}[g(x)=1]$.

Assume that we know that the distribution $D$ is not too far from uniform. Concretely, that $D(x) \le K 2^{-n}$ for some parameter $K$. This implies that 
$$
\Pr_{x \sim D}[f(x)=1] \le \Pr_{x \sim D}[g(x)=1] + \gamma K.
$$
We will choose $\gamma=1/12K$. We may assume that $\Pr_{x \sim D}[f(x)=1] \in [1/3,2/3]$, otherwise the constant $1$ function correlates with $f$ under $D$. Thus $\Pr_{x \sim D}[g(x)=1] \in [1/4,3/4]$. This implies, by the same argument as in the original paper of Jackson, there there is a term $C$ of $g$ which is $\Omega(1 / s)$-correlated with $g$. One can verify that as $g(x) \le f(x)$, $C$ is also $\Omega(1 / s)$-correlated with $f$. 

Finally, we need to bound $K$. It is known (see for example \cite{klivans2003boosting}) that boosting algorithms can be restricted to have $K=\eps^{-O(1)}$, which completes the proof.
\end{proof}

\subsection{Proof overview}

We give a high-level overview of the proof of \Cref{thm:main}. Let $L=((C_i,v_i))$ be a decision list of width $w$ and size $m$.

\paragraph{General Framework.}
Given a subset $J \subset [m]$, we denote by $L|_J$ the decision list restricted to the rules in $J$, where we delete the rest. Our goal is to find a small subset $J \subset [m]$ such that $L|_J$ approximates $L$. We say that a rule $(C_i,v_i)$ of $L$ is \emph{hit} by an input $x$ if $C_i(x)=1$ and $C_j(x)=0$ for $j<i$; in this case, $L(x)=v_i$. The main intuition underlying our approach is:
\begin{center}\it
If a rule is rarely hit by random inputs, then we can safely remove it.
\end{center}
Armed with this intuition, our approach is to choose $J$ to be the set of rules with the highest probability of being hit. We show that in order to get an $\eps$-approximation, it suffices to keep the top $\pbra{2+\frac1w\log\frac1\eps}^{O(w)}$ rules.

Our general approach follows that of Lovett and Zhang \cite{LovettZ19}. They combined two central results in the analysis of boolean functions: \emph{random restrictions} and \emph{noise stability}. The main innovation in the current work is that we apply random restrictions that fix only a small fraction of the inputs; this is in contrast to the common use of random restrictions, such as in the proof of H{\aa}stad's switching lemma \cite{Hastad87}, where most variables are fixed. The ability to handle random restrictions which fix only a small fraction is what allows us to obtain improved bounds.

\paragraph{Mild random restrictions.}

An index $i \in [m]$ is said to be \emph{useful} if there exists an assignment $x$ such that the evaluation of $L(x)$ hits the $i$-th rule (and hence outputs $v_i$). We denote the number of useful indices in $L$ by $\usenum{L}$. This notion is natural, as we can always discard rules if no assignment hits them. The main point is that restrictions can render some rules in a decision list useless. Let $\rho$ be a random restriction that keeps each variable alive with probability $\alpha$. We show that on average, the restricted decision list $L\restriction_{\rho}$ has a small number of useful indices:
$$
\E_{\rho}\sbra{\usenum{L\restriction_\rho}}\leq \pbra{\frac{4}{1-\alpha}}^w.
$$

The proof is based on an encoding argument. Let $\rho$ be a restriction for which $L\restriction_{\rho}$ has $T$ useful indices. Let $t \in [T]$ be uniformly chosen. We construct a new restriction $\rho'$ by further restricting the variables in the $t$-th useful rule so that this rule is satisfied.
Then from $\rho'$ and some small additional information $a$, we can recover both $\rho$ and $t$. This shows that the probability of $T$ being too large is very low, as the entropy of $(\rho',a)$ is much lower than that of $(\rho,t)$.

\paragraph{Noise Stability.}
Since there is no guarantee about the value on each rule of the decision list, it is convenient to consider the following index function.
Let $L=((C_i,v_i))_{i\in[m]}$ be a decision list on $n$ variables. The index function of $L$ outputs for an input $x$ the index $i$ of the first term in $L$ satisfied by $x$. Equivalently, $\Ind L$ is given by the decision list $\Ind L =((C_i,i))_{i\in[m]}$.

We make two important definitions. What we \emph{want} to analyze are the quantities
$$
p_L(i):=\Pr_x\sbra{\Ind L(x)=i},
$$
where $x$ is taken from the uniform distribution of the input.
In particular, we want to show that there is a small set of indices $J$ such that $\sum_{i \in J} p_L(i) \ge 1-\eps$. What we \emph{can} analyze using random restrictions are the quantities
$$
q_L(\alpha,i)=\Pr_{\rho}\sbra{\text{index }i\text{ is useful in }L\restriction_\rho},
$$
since it holds that
$$
\sum_iq_L(\alpha,i)=\E_{\rho}\sbra{\usenum{L\restriction_\rho}} \le \pbra{\frac{4}{1-\alpha}}^w.
$$
We use noise stability to bridge between the two.

Let $\beta=1-\alpha$. For any $x\in\bin^n$, the noise distribution $y\sim\Ncal_\beta(x)$ is sampled by taking $\Pr\sbra{y_i=x_i}=\frac{1+\beta}2$ independently for $i\in[n]$. Consider sampling $x \in \bin^n$ uniformly and $y\sim\Ncal_\beta(x)$. We can equivalently sample the pair $(x,y)$ by first
sampling a common restriction $\rho$, where each variables stays alive with probability $\alpha$, and then sample its completion for $x$ and $y$ independently. Let
$$
\Stab_L(\beta,i):=\Pr_{x,y}\sbra{\Ind L(x)=\Ind L(y)=i}.
$$
We show that $p_L(i)$ and $q_L(\alpha,i)$ are both polynomially related, by relating them to $\Stab_L(\beta,i)$:
$$
\frac{p_L(i)^2}{q_L(1-\beta,i)}\leq\Stab_L(\beta,i)\leq p_L(i)^{\frac2{1+\beta}}.
$$
The upper bound is proven by hypercontrativity, and the lower bound by a somewhat delicate Cauchy-Schwarz inequality. This allows us to obtain that
$$
p_L(i)\leq q_L(1-\beta,i)^{\frac{1+\beta}{2\beta}}.
$$
Finally, we put everything together by optimizing the value of $\beta$.

\paragraph{Related works.}
We already discussed the works of Gopalan, Meka and Reingold \cite{GopalanMR13} and Lovett and Zhang \cite{LovettZ19} which gave weaker bounds for DNF sparsification than those in \Cref{thm:main}. 

There have been previous works studying how small-width DNFs simplify under mild random restrictions that fix a small fraction of the variables (say, $1\%$). Segerlind, Buss and Impagliazzo's work \cite{segerlind2004switching}, improved by Razborov \cite{razborov2015pseudorandom}, show that width-$w$ DNFs simplify to a decision tree of depth $2^{O(w)}$. We obtain bounds on size (namely, number of useful terms) in \Cref{thm:main}, which are better than bounds on depth. However, we only bound the first moment (that is, expected number of useful terms), while \cite{razborov2015pseudorandom} bounds higher moments as well. So to some extent, the results are incomparable. We believe that with some further work, one can improve our techniques to obtain bounds on higher moments as well (this was unnecessary for the current work).
Finally, it is also worthwhile to mention the work by the authors and Alweiss \cite{alweiss2019improved}, where mild random restrictions (of a somewhat different flavor) were used to obtain improved bounds for the sunflower lemma in combinatorics.

\paragraph{Paper Organization.}
In \Cref{sec:upper}, we prove the upper bound on decision list compression. In \Cref{sec:lower}, we give the lower bounds to show the tightness of our result. 

\paragraph{Acknowledgements.}
We thank Ben Rossman for invaluable discussions. We also thank Ryan Alweiss and the anonymous reviewers for helpful suggestions on an earlier version of this paper.

\section{Upper bounds}
\label{sec:upper}

We start by make some definitions formal. We denote $[n]=\cbra{1,2,\ldots,n}$, variables are $x_1,\ldots,x_n$, and literals are $x_1,\neg x_1,\ldots,x_n,\neg x_n$. A \emph{term} is a conjunction of literals.

\begin{definition}[Decision list]
A width-$w$ size-$m$ decision list is a list $L=((C_i,v_i))_{i\in[m]}$ of rules. A rule is a pair $(C_i,v_i)$, where $C_i$ is a term containing at most $w$ literals and each $v_i$ is a value in some finite set $V$. We assume $C_m \equiv 1$, and $(C_m,v_m)$ is the final default rule.

For any $J\subseteq[m]$ with $m\in J$, we denote by  $L|_J=((C_{j},v_{j}))_{j \in J}$ the restriction of $L$ to the rules in $J$, where elements of $J$ are taken in ascending order.
\end{definition}

The evaluation of $L$ given assignment $x$ is to find the first index $i$ such that $C_i(x)=1$ and then to output $L(x)=v_i$.
We make additional remarks for the decision list to avoid potential pitfalls.
\begin{itemize}
\item If $m\notin J$, we will consider $L|_J$ invalid, as it does not have a default rule at the end.
\item No variable appears in any single term more than once, which rules out $x_1\land x_1$ and $x_1\land\neg x_1$.
\end{itemize}

Our goal in this section is to prove the following theorem, which is the upper bound part in \Cref{thm:main}.

\begin{theorem}\label{thm:compress}
Let $L=((C_i,v_i))_{i\in[m]}$ be a width-$w$ decision list. Then for every $\eps>0$, there exists $J\subseteq[m],m\in J$ of size
$|J|=\pbra{2+\frac1w\log\frac1\eps}^{O(w)}$
such that $\Pr\sbra{L(x)\neq L|_J(x)}\leq\eps$.
\end{theorem}

\subsection{Useful indices}
Since there is no guarantee about the value on each rule of the decision list, it is convenient to consider the index function.
Let $L=((C_i,v_i))_{i\in[m]}$ be a decision list on $n$ variables. The index function of $L$ is a function $\Ind L:\bin^n\to[m]$, given by
$$
\Ind L(x)=\min\cbra{i\in[m]\mid C_i(x)=1}.
$$
Equivalently, $\Ind L$ is given by the decision list $\Ind L =((C_i,i))_{i\in[m]}$. Using the index function, it suffices to discard some rules of $L$ and show it still approximates the index function.

\begin{claim}\label{clm:transfertoindex}
Let $L=((C_i,v_i))_{i\in[m]}$ be a decision list. Then for any $J\subseteq[m],m\in J$, we have
$$
\Pr\sbra{L(x)\neq L|_J(x)}\leq\Pr\sbra{\Ind L(x)\notin J}.
$$
\end{claim}
\begin{proof}
This follows as if $\Ind L(x)=j\in J$, then $L(x)=L|_J(x)=v_j$.
\end{proof}

Obviously, if a rule of a decision list is covered by some previous rules, then we can safely remove it. For example, in $(x_1,1),(x_1\land x_2,2)$ the second rule is useless. To make this more formal, we introduce the following notion of a \emph{useful index}.

\begin{definition}[Useful index]
Given size-$m$ decision list $L$, an index $i\in[m]$ is said to be \emph{useful} if there exists an assignment $x$ such that $\Ind L(x)=i$.
We denote by $\usenum L$ the number of useful indices in $L$.
\end{definition}

\begin{example}
Assume $L=((x_1,a),(x_1\land\neg x_2,b),(1,c),(x_1,d),(1,e))$. Then indices $1,3$ are useful, but indices $2,4,5$ are not. So $\usenum L=2$.
\end{example}

The main intuition underlying our approach is that rules that are hardly hit by random inputs can be removed. Motivated by this, we define \emph{hit probability}
$$
p_L(i):=\Pr\sbra{\Ind L(x)=i}.
$$

\begin{claim}\label{clm:hitprobsum}
For any size-$m$ decision list $L$, we have $\sum_{i=1}^mp_L(i)=1$.
\end{claim}
\begin{proof}
This follows as the events $[\Ind L(x)=i]$ are a partition of the probability space.
\end{proof}

The following is our main technical lemma.

\begin{lemma}\label{lem:compress}
Let $L=((C_i,v_i))_{i\in[m]}$ be a width-$w$ decision list. Sort $[m]=\{j_1,\ldots,j_m\}$ such that
$p_L(j_1)\geq p_L(j_2)\geq\cdots\geq p_L(j_m)$. For any $\eps>0$, let
$$
t=\pbra{2+\frac1w\log\frac1\eps}^{O(w)}.
$$
Then for $J=\{j_1,\ldots,j_t,m\}$ it holds that $\Pr\sbra{\Ind L(x)\notin J}\leq\eps$.
\end{lemma}

The proof of \Cref{thm:compress} follows immediately, by combining \Cref{lem:compress} and \Cref{clm:transfertoindex}.

\subsection{Random restrictions and encoding}

A \emph{restriction} on $n$ variables is $\rho\in\cbra{0,1,*}^n$.
An $(n,k)$-random restriction is the uniform distribution over restrictions $\rho\in\cbra{0,1,*}^n$ with exactly $k$ stars, which we denote by $\Rscr(n,k)$.
An $(n,\alpha)$-random restriction, which we denote by $\Uscr(n,\alpha)$, assigns independently each bit of the restriction $\rho$ to $0,1,*$ with probability $\frac{1-\alpha}2,\frac{1-\alpha}2,\alpha$ respectively.
Given a decision list $L:\bin^n\to V$, its restriction under $\rho$ is $L\restriction_\rho:\bin^{\rho^{-1}(*)}\to V$.

\begin{definition}[Useful probability]
Given size-$m$ decision list $L$ and $\alpha\in(0,1)$, the useful probability of an index $i\in[m]$ is
$$
q_L(\alpha,i):=\Pr_{\rho\sim\Uscr(n,\alpha)}\sbra{\text{index }i\text{ is useful in }L\restriction_\rho}.
$$
\end{definition}

Note that we assume $L$ initially does not contain useless rules, so for any $\alpha$ and $i$, we always have $q_L(\alpha,i)>0$.
We also have the following simple fact regarding useful probability.

\begin{claim}\label{clm:usefulprobsum}
For any size-$m$ decision list $L$, we have $\sum_{i=1}^mq_L(\alpha,i)=\E_{\rho\sim\Uscr(n,\alpha)}\sbra{\usenum{L\restriction_\rho}}$.
\end{claim}
\begin{proof}
Let $1_{\rho,i}$ be the indicator of index $i$ being useful in $L\restriction_\rho$. Then
$$
\E_{\rho\sim\Uscr(n,\alpha)}\sbra{\usenum{L\restriction_\rho}}=\E_\rho\sbra{\sum_{i=1}^m 1_{\rho,i}}=\sum_{i=1}^m\E_\rho\sbra{1_{\rho,i}}=\sum_{i=1}^m q_L(\alpha,i).
$$
\end{proof}

Now we present an encoding/decoding scheme for random restriction and analyze the expectation in \Cref{clm:usefulprobsum} explicitly. Let $\alpha \in (0,1)$ be such that $\alpha n$ is an integer. Define:
\begin{align*}
&\Ucal:=\cbra{(\rho,s)\bigg|\rho\in\Rscr(n,\alpha n),s \in \{1,\ldots,\usenum{L\restriction_\rho}\}}\\
&\Vcal:=\cbra{(\rho',a)\bigg|\rho'\in\bigcup_{k=0}^w\Rscr(n,\alpha n-k),a\in\cbra{\textsc{Old},\textsc{New}}^w}.
\end{align*}

We define two deterministic algorithms $\Enc:\Ucal\to\Vcal$ and $\Dec:\Enc(\Ucal)\subseteq\Vcal\to\Ucal$ such that $\Dec(\Enc(\rho,s))=(\rho,s)$ holds for any $(\rho,s)\in\Ucal$.

\begin{algorithm}[H]
    \caption{Encoding algorithm $\Enc(\rho,s)$}\label{alg:singleEnc}
    \DontPrintSemicolon
    \LinesNumbered
    \KwIn{restriction and index $(\rho,s)\in\Ucal$}
    \KwOut{restriction and string $(\rho',a)\in\Vcal$}
    $I\gets\cbra{i\mid i\text{ is a useful index in }L\restriction_\rho}$\;
    $j\gets\text{the $s$-th element in }I$\;
    $\rho'\gets\rho,a\gets\varnothing$\;
    \tcc{Assume $C_j=\bigwedge_{k=1}^cy_{j_k},y_{j_k}\in\cbra{x_{j_k},\neg x_{j_k}},c\leq w$}
    \For{$k=1$ \KwTo $c$}{
        \eIf{$\rho(x_{j_k})\in\bin$}{
            Append $a$ with $\textsc{Old}$\tcc*{$x_{j_k}$ is already set by $\rho$}
        }{
            Append $a$ with $\textsc{New}$\tcc*{$x_{j_k}$ is newly set to satisfy this term}
            \leIf{$y_{j_k}=x_{j_k}$}{Update $\rho'(x_{j_k})\gets1$}{Update $\rho'(x_{j_k})\gets0$}
        }
        Complete $a$ arbitrarily to length $w$\;
    }
    \Return{$(\rho',a)$}
\end{algorithm}

\begin{algorithm}[ht]
    \caption{Decoding algorithm $\Dec(\rho',a)$}\label{alg:singleDec}
    \DontPrintSemicolon
    \LinesNumbered
    \KwIn{restriction and string $(\rho',a)\in\Enc(\Ucal)\subseteq\Vcal$}
    \KwOut{restriction and index $(\rho,s)\in\Ucal$}
    $j\gets\text{index of the first satisfied term in }L\restriction_{\rho'}$\;
    $\rho\gets\rho'$\;
    \tcc{Assume $C_j=\bigwedge_{k=1}^cy_{j_k},y_{j_k}\in\cbra{x_{j_k},\neg x_{j_k}},c\leq w$}
    \For{$k=1$ \KwTo $c$}{
        \If(\tcc*[f]{$x_{j_k}$ was not set by $\rho$}){$a_k=\textsc{New}$}{
            Update $\rho(x_{j_k})\gets*$\;
        }
    }
    $I\gets\cbra{i\mid i\text{ is a useful index in }L\restriction_\rho}$\;
    $s\gets\text{rank of $j$ in $I$}$\;
    \Return{$(\rho,s)$}
\end{algorithm}

The following claim proves the correctness of the encoding and decoding algorithms.

\begin{claim}\label{claim:singleEncDec}
$\Dec(\Enc(\rho,s))=(\rho,s)$ holds for any $(\rho,s)\in\Ucal$.
\end{claim}

\begin{proof}

Sort literals in each term of $L=((C_i,v_i))_{i\in[m]}$ arbitrarily.
To justify the correctness, let $(\rho',a)=\Enc(\rho,s)$, then we need to ensure:
\begin{itemize}
\item $\Dec(\rho',a)$ obtains the same $j$ in line 1 as $\Enc(\rho,s)$ does in line 2:

During $\Enc(\rho,s)$, index $j$ is useful in $L\restriction_\rho$, thus setting unfixed variables to satisfy $C_j$ will not make any term $C_i$ for $i<j$ satisfied. Hence the first satisfied term in $L\restriction_{\rho'}$ is $C_j$.

\item $\Dec(\rho',a)$ in line 8 obtains the correct $\rho$:

Since each term is sorted in advance, and $a$ encodes which variable in $C_j$ is set by $\Enc(\rho,s)$ rather than $\rho$, the loop in $\Dec(\rho',a)$ will set these variables back to $*$ and recover $\rho$.
\end{itemize}
\end{proof}

\begin{corollary}
\label{cor:UV}
$|\Ucal| \le |\Vcal|$.
\end{corollary}

\begin{proof}
$\Enc$ is an injection from $\Ucal$ to $\Enc(\Ucal) \subset \Vcal$.
\end{proof}

\begin{lemma}\label{lemma:expectedusenum}
Let $L$ be a width-$w$ decision list on $n$ variables and let $\alpha\in(0,1)$. Then
$$
\E_{\rho\sim\Uscr(n,\alpha)}\sbra{\usenum{L\restriction_\rho}}\leq\pbra{\frac4{1-\alpha}}^w.
$$
\end{lemma}

\begin{proof}
We first prove the bound for $\rho\sim\Rscr(n,\alpha n)$ and then increase the number of variables to infinity, by adding dummy variables. This proves the desired bound as for $n' \to \infty$, the restriction of $\Rscr(n',\alpha n')$ to the first $n$ variables converges to $\Uscr(n,\alpha)$. We have
\begin{align*}
&\E_{\rho\sim\Rscr(n,\alpha n)}\sbra{\usenum{L\restriction_\rho}}=\frac1{\abs{\Rscr(n,\alpha n)}}\sum_{\rho\in\Rscr(n,\alpha n)}\usenum{L\restriction_\rho}\\
=&\frac{|\Ucal|}{\abs{\Rscr(n,\alpha n)}}\leq\frac{|\Vcal|}{\abs{\Rscr(n,\alpha n)}}\leq\frac{\pbra{\sum_{k=0}^w\binom{n}{\alpha n-k}2^{(1-\alpha)n+k}}\times2^w}{\binom n{\alpha n}2^{(1-\alpha)n}}\\
\leq&\frac{\pbra{\sum_{k=0}^w\binom{n}{\alpha n-k}}\times4^w}{\binom n{\alpha n}}\leq\frac{\binom{n+w}{\alpha n}\times4^w}{\binom n{\alpha n}}\leq\pbra{\frac4{1-\alpha}}^w.
\end{align*}
\end{proof}

\subsection{Noise stability}

We use noise stability as a bridge between $p_L(i)$ and $q_L(\alpha,i)$.

\begin{definition}[Noisy distribution]
Given $x\in\bin^n$ and a noise parameter $\beta\in(0,1)$, we denote by $\Ncal_\beta(x)$ the distribution over $y\in\bin^n$, where $\Pr\sbra{y_i=x_i}=\frac{1+\beta}2,\Pr\sbra{y_i\neq x_i}=\frac{1-\beta}2$ independently for all $i\in[n]$.
\end{definition}

\begin{definition}[Stability]
Let $g:\bin^n\to\bin$ be a boolean function. The $\beta$-stability of $g$ is
$$
\Stab_\beta(g)=\Pr_{x \in \{0,1\}^n, y\sim\Ncal_\beta(x)}\sbra{g(x)=g(y)=1}.
$$
\end{definition}

The hypercontractive inequality (see for example \cite{o2014analysis}, page 259) allows us to bound the stability of a boolean function by its acceptance rate.

\begin{fact}\label{fct:hypercontractivity}
Let $g:\bin^n \to \bin$ and $\beta\in(0,1)$. Then $\Stab_\beta(g)\leq\pbra{\Pr\sbra{g(x)=1}}^{\frac2{1+\beta}}$.
\end{fact}

Next, we define index stability and relate it to useful probability $q_L(\cdot,\cdot)$ and hit probability $p_L(\cdot)$.

\begin{definition}[Index stability]
Given a size-$m$ decision list $L$ on $n$ variables, the $\beta$-stability of index $i\in[m]$ is
$$
\Stab_L(\beta,i):=\Pr_{x \in \{0,1\}^n, y\sim\Ncal_\beta(x)}\sbra{\Ind L(x)=\Ind L(y)=i}.
$$
\end{definition}

\begin{lemma}[Bridging lemma]\label{lem:bridging}
Let $L$ be a size-$m$ width-$w$ decision list on $n$ variables. Then for any index $i\in[m]$ and $\beta\in(0,1)$, we have
$$
\frac{p_L(i)^2}{q_L(1-\beta,i)}\leq\Stab_L(\beta,i)\leq p_L(i)^{\frac2{1+\beta}}.
$$
\end{lemma}

\begin{proof}
We first prove the upper bound. Let $g:\bin^n\to\bin$ be an indicator boolean function for $\Ind L(x)=i$. Then using \Cref{fct:hypercontractivity}, we have
$$
\Stab_L(\beta,i)=\Stab_\beta(g)\leq\pbra{\Pr\sbra{g(x)=1}}^{\frac2{1+\beta}}=\pbra{\Pr\sbra{\Ind L(x)=i}}^{\frac2{1+\beta}}=p_L(i)^{\frac2{1+\beta}}.
$$

We now turn to prove the lower bound.
Let $\alpha=1-\beta$. Observe that we can sample $(x,y)$ where $x \in \{0,1\}^n,y\sim\Ncal_\beta(x)$ as follows:
\begin{itemize}
\item Sample restriction $\rho\sim\Uscr(n,\alpha)$;
\item Sample uniform $x'\in\bin^{\rho^{-1}(*)}$ and complete stars in $\rho$ with it as $x$;
\item Sample uniform $y'\in\bin^{\rho^{-1}(*)}$ and complete stars in $\rho$ with it as $y$.
\end{itemize}
We thus have
$$
\Stab_L(\beta,i)=\Pr_{\rho,x',y'}\sbra{\Ind L\restriction_\rho(x')=\Ind L\restriction_\rho(y')=i}.
$$
We now make a seemingly redundant, but surprisingly useful, conditioning. Let $\Ecal(\rho,i)$ denote the event
$$
\Ecal(\rho,i) := \sbra{i \text{ is useful in }L\restriction_\rho}.
$$
Then we can equivalently write
$$
\Stab_L(\beta,i)=\Pr_{\rho,x',y'}\sbra{\Ind L\restriction_\rho(x')=\Ind L\restriction_\rho(y')=i \wedge \Ecal(\rho,i)}.
$$
For any fixed $\rho$, define
$$
r_{\rho}(i):=\Pr_{x'}\sbra{\Ind L\restriction_{\rho}(x')=i}.
$$
Since $x',y'$ are independent for any fixed restriction, we have
\begin{align*}
\Stab_L(\beta,i)
=&\Pr_{\rho}[\Ecal(\rho,i)] \cdot \Pr_{\rho,x',y'} \sbra{\Ind L\restriction_\rho(x')=\Ind L\restriction_\rho(y')=i\bigg|\Ecal(\rho,i)}\\
=&q_L(\alpha,i) \cdot \E_{\rho}\sbra{r_{\rho}(i)^2\bigg|\Ecal(\rho,i)}\\
\geq&q_L(\alpha,i) \cdot \pbra{\E_{\rho}\sbra{r_{\rho}(i)\bigg|\Ecal(\rho,i)}}^2\tag{Cauchy-Schwarz inequality}\\
=&\frac1{q_L(\alpha,i)}\pbra{q_L(\alpha,i) \cdot \E_{\rho}\sbra{r_{\rho}(i)\bigg|\Ecal(\rho,i)}}^2\\
=&\frac1{q_L(\alpha,i)}\pbra{\Pr_{\rho,x'}\sbra{\Ind L\restriction_\rho(x')=i\wedge\Ecal(\rho,i)}}^2\\
=&\frac1{q_L(\alpha,i)}\pbra{\Pr_{\rho,x'}\sbra{\Ind L\restriction_\rho(x')=i}}^2\\
=&\frac1{q_L(\alpha,i)}\pbra{\Pr_{x}\sbra{\Ind L(x)=i}}^2
=\frac{p_L(i)^2}{q_L(\alpha,i)}.
\end{align*}
\end{proof}

\begin{corollary}\label{cor:bridging}
Let $L$ be a size-$m$ width-$w$ decision list. Then for any index $i\in[m]$ and $\beta\in(0,1)$, we have
$$
p_L(i)\leq q_L(1-\beta,i)^{\frac{1+\beta}{2\beta}}.
$$
\end{corollary}

As a remark, we note that \Cref{lem:bridging} can be generalized to arbitrary boolean functions with a similar proof.

\begin{lemma}\label{lem:general_bridging}
Let $g:\bin^n\to\bin$ be a boolean function which is not identically zero. Set $|g|=\Pr\sbra{g(x)=1}$. Then for any $\beta\in(0,1)$, we have
$$
\frac{|g|^2}{\displaystyle\Pr_{\rho\sim\Uscr(n,1-\beta)}[g\restriction_\rho\not\equiv0]}\leq\Stab_\beta(g)\leq|g|^{\frac2{1+\beta}}.
$$
\end{lemma}

\subsection{Putting everything together}

Now we put everything together and give the proof of \Cref{lem:compress}.
\begin{proof}[Proof of \Cref{lem:compress}]
Recall that we sorted $[m]=\{j_1,\ldots,j_m\}$ such that
$p_L(j_1)\geq p_L(j_2)\geq\cdots\geq p_L(j_m)$. Let $J=\{j_1,\ldots,j_t,m\}$ for $t$ to be optimized later.

Next, let $\beta\in(0,1)$ to be optimized later and set $\alpha=1-\beta$. Sort $[m]=\{i_1,\ldots,i_m\}$ such that
$q_L(\alpha,i_1)\geq q_L(\alpha,i_2)\geq\cdots\geq q_L(\alpha,i_m)$.
By \Cref{clm:usefulprobsum} and \Cref{lemma:expectedusenum}, we have
$$
\sum_{k=1}^mq_L(\alpha,i_k)=\E_{\rho\sim\Uscr(n,\alpha)}\sbra{\usenum{L\restriction_\rho}}\leq\pbra{\frac4{1-\alpha}}^w=\pbra{\frac4\beta}^w.
$$
Note that we have sorted $q_L$ in decreasing order, so
$$
q_L(\alpha,i_k)\leq\frac1k\pbra{\frac4\beta}^w.
$$
Observe that $j_1,\ldots,j_t$ have the largest hit probability, and apply \Cref{cor:bridging}, then
\begin{align*}
\sum_{j\notin J}p_L(j)&\leq\sum_{k=t+1}^mp_L(j_k)\leq\sum_{k=t+1}^mp_L(i_k)\leq
\sum_{k=t+1}^m q_L(\alpha,i_k)^{\frac{1+\beta}{2\beta}}\\
&\leq\pbra{\frac4\beta}^{w\times\frac{1+\beta}{2\beta}}\sum_{k \ge t+1}\pbra{\frac1k}^{\frac{1+\beta}{2\beta}}\\
&\leq\pbra{\frac4\beta}^{w\times\frac{1+\beta}{2\beta}}\times\frac{2\beta}{1-\beta}\times t^{-\frac{1-\beta}{2\beta}}.
\end{align*}
If we restrict $\beta\leq1/2$ and choose
$$
t=\pbra{\frac1\eps}^{\frac{2\beta}{1-\beta}}\pbra{\frac4\beta}^{w\times\frac{1+\beta}{1-\beta}}\pbra{\frac{2\beta}{1-\beta}}^{\frac{2\beta}{1-\beta}}\leq4\pbra{\frac1\eps}^{4\beta}\pbra{\frac4\beta}^{3w},
$$
then
$$
\Pr\sbra{\Ind L(x)\notin J}=\sum_{j\notin J}p_L(j)\leq\eps.
$$
Now we divide $\eps$ into two cases. Assume $\eps=2^{-\ell w}$. Then:
\begin{itemize}
\item If $\ell \le 2$ we set $\beta=1/2$ and get $t=2^{O(w)}$.
\item If $\ell \ge 2$ we set $\beta=1/\ell$ and get $t=\ell^{O(w)}$.
\end{itemize}
One can verify that in either case we get
$$
t=\pbra{2+\frac1w\log\frac1\eps}^{O(w)}.
$$
\end{proof}

\section{Lower bounds}
\label{sec:lower}

In this section, we prove two lower bounds for decision list compression, which show that the bounds in \Cref{thm:main} are tight up to constants.

\begin{claim}\label{lem:approxlowerbound}
For any $w$, there is a width-$w$ decision list $L:\bin^w\to\bin$ such that
$$
\Pr\sbra{L(x)\neq L'(x)}>1/3
$$
for any width-$w$ decision list $L'$ of size at most $2^w/100w$.
\end{claim}
\begin{proof}
Since any boolean function on $w$ variables can be expressed as some width-$w$ decision list, there are $2^{2^w}$ possible $L$. On the other hand, for any fixed $L'$, it can approximate at most
$$
\binom{2^w}{2^w/3}\times2^{2^w/3}\leq2^{0.97\times2^w}
$$
different boolean functions within distance $1/3$; and for fixed size $m$, there are at most $\pbra{3^w\times 2}^m$ distinct size-$m$ width-$w$ decision lists. As small-size decision lists can be embedded in larger ones, when restricted to size at most $2^w/100w$, width-$w$ decision lists only approximate at most
$$
\pbra{3^w\times 2}^{\frac{2^w}{100w}}\times2^{0.97\times2^w}<2^{2^w}
$$
different boolean functions on $w$ variables.
\end{proof}

\begin{claim}\label{lem:exactlowerbound}
For any $w$ and $n>2w$, there is a width-$w$ decision list $L:\bin^n\to\bin$ which is not equivalent to any width-$w$ decision list $L'$ of size smaller than $\binom nw/n^2$.
\end{claim}
\begin{proof}
Let $m=\binom nw$ and sort all $\binom nw$ subsets of $[n]$ with size $w$ as $\cbra{S_1,\ldots,S_m}$ arbitrarily.
For any $i\in[m]$, define $C_i=\bigwedge_{j\in S_i}x_j$.
For any $v\in\bin^m$, let $L_v=((C_1,v_1),\ldots,(C_m,v_m),(1,0))$ be a size-$(m+1)$ width-$w$ decision list.

As small-size decision lists can be embedded in larger ones, assume towards a contradiction that any $L_v$ is equivalent to some size-$(m/n^2)$ width-$w$ decision list $L_v'$. Given $L_v'$, we can recover $L_v$ by enumerating all assignments, since all rules in $L_v$ are useful. Thus, by counting argument, the number of possible $L_v'$ is upper bounded by
$$
\pbra{2\times\sum_{k=0}^w2^k\binom nk}^{\binom nw/n^2}\leq\binom nw^{2m/n^2}<2^m.
$$
\end{proof}

Now the general lower bound follows immediately.

\begin{corollary}\label{cor:lowerbound}
For any $w$ and $\eps\leq1/3$, there is a width-$w$ decision list $L$ such that
$$
\Pr\sbra{L(x)\neq L'(x)}>\eps
$$
holds for any width-$w$ decision list $L'$ of size at most
$$
\pbra{2+\frac1w\log\frac1\eps}^{O(w)}.
$$
\end{corollary}
\begin{proof}
For $\eps\geq2^{-2w}$, let $L$ be the decision list in \Cref{lem:approxlowerbound}. Then it cannot be approximated within $\eps<1/3$ by a decision list $L'$ of size at most
$$
\frac{2^w}{100w}=\pbra{2+\frac1w\log\frac1\eps}^{O(w)}.
$$
For $\eps<2^{-2w}$, let $L$ be the decision list in \Cref{lem:exactlowerbound} with $n=\log(1/\eps)$. Since now $\eps=2^{-n}$, the desired $L'$ must be equivalent to $L$. Thus it cannot be realized by a decision list $L'$ of size at most
$$
\frac{\binom nw}{n^2}=\binom{\log\frac1\eps}w^{O(1)}=\pbra{2+\frac1w\log\frac1\eps}^{O(w)}.
$$
\end{proof}

\bibliographystyle{abbrv}
\bibliography{ref}

\end{document}